\documentclass[11pt,reqno]{amsart}

\usepackage[T1]{fontenc}
\usepackage{lmodern}
\usepackage{microtype}
\usepackage{amsmath,amssymb,amsthm,mathtools}
\usepackage{mathrsfs}
\usepackage{enumitem}
\usepackage{comment}
\usepackage[colorlinks=true,linkcolor=blue,citecolor=blue,urlcolor=blue]{hyperref}

\usepackage{mathtools}
\mathtoolsset{showonlyrefs}

\hypersetup{
  pdftitle={Black hole uniqueness via singular harmonic maps},
  pdfauthor={Qing Han, Marcus Khuri, Gilbert Weinstein, and Jingang Xiong},
  pdfsubject={Black hole uniqueness and logarithmic angle defects for singular harmonic maps},
  pdfkeywords={black hole uniqueness, stationary vacuum black holes, singular harmonic maps, horizon rods, logarithmic angle defects, maximum principle}
}

\allowdisplaybreaks
\numberwithin{equation}{section}

\newtheorem{theorem}{Theorem}[section]
\newtheorem{proposition}[theorem]{Proposition}
\newtheorem{lemma}[theorem]{Lemma}
\newtheorem{corollary}[theorem]{Corollary}

\newtheorem{remark}[theorem]{Remark}

\newcommand{\R}{\mathbb R}
\newcommand{\Hh}{\mathbb H}

\newcommand{\dd}{\,d}
\newcommand{\cJ}{\mathcal J}
\newcommand{\cZ}{\mathcal Z}

\newcommand{\eps}{\varepsilon}
\newcommand{\grad}{\nabla}
\newcommand{\sgn}{\operatorname{sgn}}

\newcommand{\SO}{\operatorname{SO}}

\title[Black Hole Uniqueness]{Kerr Black Hole Uniqueness}

\author[Q. Han]{Qing Han}
\address{Department of Mathematics, University of Notre Dame, Notre Dame, IN 46556, USA}
\email{qhan@nd.edu}

\author[M. Khuri]{Marcus Khuri}
\address{Department of Mathematics, Stony Brook University, Stony Brook, NY 11794, USA}
\email{marcus.khuri@stonybrook.edu}

\author[G. Weinstein]{Gilbert Weinstein}
\address{Department of Mathematics and Department of Physics, Ariel University, Ariel, 40700, Israel}
\email{gilbertw@ariel.ac.il}

\author[J. Xiong]{Jingang Xiong}
\address{School of Mathematical Sciences, Laboratory of Mathematics and Complex Systems, MOE, Beijing Normal University, Beijing 100875, China}
\email{jx@bnu.edu.cn}

\thanks{Q. Han acknowledges the support of NSF Grant DMS-2305038. M. Khuri acknowledges the support of NSF Grant DMS-2405045. J. Xiong acknowledges the partial support of NSFC Grants 12325104. G. Weinstein acknowledges the support of ISF Grant 1119/26.}

\begin{document}

\begin{abstract}
We prove the black hole uniqueness conjecture in the axially symmetric, stationary vacuum setting: there is no regular, asymptotically flat equilibrium configuration with more than one horizon component. More precisely, we establish that in every such multi-horizon configuration, the logarithmic angle defect along every bounded axis rod is negative, which implies that the net interaction force across each such segment is attractive. The proof is based on a refined asymptotic analysis of the associated singular harmonic maps, partially drawn from the companion work \cite{HKWXAsymptotic}, together with a global maximum principle bound for the Weyl conformal factor arising from a novel scalar curvature related differential inequality.
\end{abstract}

\maketitle

\section{Introduction}
\label{sec1}

The black hole uniqueness problem asks whether every regular, stationary vacuum black hole spacetime must belong to the two-parameter Kerr family of solutions, completely characterized by its mass and angular momentum measured at infinity. This problem naturally divides into two stages. The first is the rigidity problem: proving that any stationary black hole spacetime necessarily possesses an additional, rotational Killing symmetry which commutes with the time symmetry. Hawking's rigidity theorem establishes this under the hypothesis of real analyticity. Major progress toward removing analyticity was achieved by Alexakis--Ionescu--Klainerman, who established rigidity under various assumptions \cite{AIK2010a}, including closeness to Kerr \cite{AIK2010b} and small angular momentum \cite{AIK2014}. The second stage is the classification problem within the stationary and axisymmetric class. When the event horizon is connected, classical uniqueness theorems \cite{Carter1973,Robinson1975,Mazur1982} identify the spacetime uniquely as a member of the Kerr family.

In this paper, we resolve the long-standing multi-horizon balance problem in the stationary and axisymmetric vacuum setting. Specifically, we prove a strict sign theorem for the logarithmic angle defects: every bounded axis segment separating two consecutive horizon components has a negative logarithmic angle defect. Equivalently, the interaction force supported on each such bounded axis component is attractive. As a consequence, any stationary, axisymmetric vacuum spacetime containing more than one horizon component necessarily exhibits a conical singularity along every bounded axis segment, and therefore cannot represent a regular equilibrium configuration. This establishes the nonexistence of multi-black hole vacuum equilibria. Our result is stronger than the minimal requirement for uniqueness, which only requires that at least one bounded axis segment carry a conical singularity.

The classical Bach--Weyl analysis \cite{BachWeyl1922,BachWeyl2012Reprint} demonstrated that in the static two-body setting, a conical singularity, interpreted physically as a strut, must appear on the axis between the bodies. In the rotating setting, Weinstein's singular harmonic map construction \cite{Weinstein1990,Weinstein1992} produces, for any finite number of nondegenerate horizons, an axisymmetric stationary vacuum metric that is smooth away from possible conical singularities on the axis. An early nonexistence result under an additional equatorial symmetry hypothesis was obtained by Li and Tian \cite{LiTian1991}. The axis and pole regularity theory and the associated force analysis were systematically developed in \cite{Weinstein1990,LiTian1992,LiTian1993,Weinstein1994}. For two horizon components, the nonexistence of equilibrium was proven by Neugebauer and Hennig in the subextreme case \cite{NeugebauerHennig2009} and extended to the extreme case in \cite{HennigNeugebauer2011}; see also the two-Kerr singularity analysis of Chru\'sciel, Eckstein, Nguyen, and Szybka \cite{CENS2011}. They relied on the complete integrability of the Ernst system and the inverse scattering method to construct explicit solutions. That technique faces insurmountable complexities for $N\geq 3$ horizon components. The nonexistence problem for an arbitrary finite number of horizon components has remained open.

We now describe the geometric and analytic framework. Let $\Gamma=\{\rho=0\}\subset\R^3$ denote the rotation axis in cylindrical coordinates $(\rho,z,\phi)$, with $\phi\in [0,2\pi)$. On the domain of outer communications away from the orbit space boundary, the spacetimes under consideration have topology $\mathbb{R}\times (\mathbb{R}^3 \setminus\Gamma)$, and their stationary axisymmetric vacuum metrics take the Weyl--Papapetrou form
\begin{equation}\label{eq:intro-WP-metric}
\mathbf g=-e^{2U}dt^2
+\rho^2e^{-2U}(d\phi+w\,dt)^2
+e^{-2U+2\alpha}(d\rho^2+dz^2).
\end{equation}
Here $\partial_t$ and $\partial_\phi$ denote the stationary and rotational Killing fields, respectively. Setting
\[
u=U-\ln\rho,
\]
the vacuum Einstein equations reduce to the axisymmetric harmonic map system
\begin{equation}\label{eq:intro-harmonic-map-system}
\Delta u=2e^{4u}|\grad v|^2,
\quad
\Delta v+4\grad u\cdot\grad v=0,
\end{equation}
for a map
\[
\Phi=(u,v):\R^3\setminus\Gamma\longrightarrow\Hh^2,
\quad
ds_{\Hh^2}^2=du^2+e^{4u}dv^2,
\]
where $v$ is the twist potential and $\Delta$ is the standard Euclidean Laplacian on $\R^3$. Once $\Phi$ is determined, the remaining metric potentials are recovered, up to additive constants, from
\begin{equation}\label{eq:intro-w-quadrature}
w_\rho=2\rho e^{4u}v_z,\quad
w_z=-2\rho e^{4u}v_\rho,
\end{equation}
and
\begin{equation}\label{eq:intro-alpha-quadrature}
\alpha_\rho=\rho\left[U_\rho^2-U_z^2+e^{4u}(v_\rho^2-v_z^2)\right],\quad
\alpha_z=2\rho\left[U_\rho U_z+e^{4u}v_\rho v_z\right].
\end{equation}
The integrability conditions for \eqref{eq:intro-w-quadrature} and \eqref{eq:intro-alpha-quadrature} follow from \eqref{eq:intro-harmonic-map-system}. The asymptotically flat normalization fixes the gauge constants uniquely by requiring $w\to0$ and $\alpha\to0$ as $r=\sqrt{\rho^2+z^2}\to\infty$. Consequently, the rod structure, horizon angular momenta, logarithmic angle defects, and interaction forces are entirely determined by the singular harmonic map \cite{Carter1973,Weinstein1990}.

The boundary $\{\rho=0\}$ of the orbit half-plane decomposes into an alternating sequence of horizon and axis intervals. Let
\begin{equation}\label{eq:intro-horizon-components}
\mathcal H_i=[z_i^- ,z_i^+],
\quad
z_i^-\leq z_i^+<z_{i+1}^-,
\quad 1\leq i\leq N,
\end{equation}
be the ordered horizon rods. If $z_i^-<z_i^+$, the interior of $\mathcal H_i$ corresponds to a nondegenerate horizon rod. The Killing field
\[
\partial_t+\Omega_i\partial_\phi,
\quad
\Omega_i=-w|_{\mathrm{Int}\,\mathcal H_i},
\]
is null on this component, where the constant $\Omega_i$ represents the horizon angular velocity, and the endpoints $(0,z_i^-)$ and $(0,z_i^+)$ correspond to the south and north poles of the horizon 2-sphere. If $z_i^-=z_i^+=z_i$, the horizon is degenerate (extremal) and is represented in the orbit space by the puncture $p_i=(0,z_i)\in\R^3$. The complementary open intervals
\begin{equation}\label{eq:intro-rods}
\Gamma_1=(-\infty,z_1^-),\quad
\Gamma_{j+1}=(z_j^+ , z_{j+1}^-),\quad 
\Gamma_{N+1}=(z_N^+ ,\infty),
\quad 1\leq j\leq N-1,
\end{equation}
form the axis rods. On each $\Gamma_j$, the rotational Killing field $\partial_\phi$ vanishes, and the twist potential is constant:
\begin{equation}\label{eq:intro-twist-potential-constant}
v=c_j \quad \text{on }\Gamma_j.
\end{equation}
Under the orientation convention of \cite{HKWXAsymptotic,HKWXMass}, the Komar angular momentum of the $i$th horizon component is given by the jump
\begin{equation}\label{eq:intro-angular-momentum-jump}
\cJ_i=\frac{c_{i+1}-c_i}{4}.
\end{equation}
An illustrative configuration containing one degenerate horizon and two nondegenerate horizons is shown in Figure~\ref{conf}.

\begin{figure}[htbp]
\centering
\includegraphics[height=8.5cm]{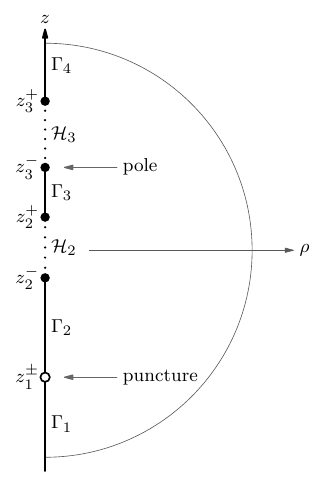}
\caption{Orbit space rod structure for a configuration with one degenerate horizon puncture ($z_1$) and two nondegenerate horizon rods ($\mathcal{H}_1, \mathcal{H}_3$).}
\label{conf}
\end{figure}

The \textit{logarithmic angle defects} $\mathbf{b}_j$ are geometric invariants of the reconstructed metric. Open axis regularity combined with \eqref{eq:intro-alpha-quadrature} implies that $\alpha$ assumes a constant value on each axis rod. For any $z_0\in\Gamma_j$, computing the ratio of circumference to radius from \eqref{eq:intro-WP-metric} yields
\begin{equation}\label{eq:intro-defect-alpha-trace}
e^{\mathbf b_j}
:=\lim_{\rho\rightarrow 0}
\frac{\displaystyle\int_0^\rho e^{\alpha(s,z_0)-U(s,z_0)}\,ds}{\rho e^{-U(\rho,z_0)}}
=e^{\alpha(0,z_0)},
\end{equation}
so that $\mathbf b_j=\alpha|_{\Gamma_j}$. The corresponding interaction force \cite{Weinstein1990} is given by
\begin{equation}\label{eq:intro-force}
\mathcal F_j=\frac14\left(e^{-\mathbf b_j}-1\right).
\end{equation}
Hence, $\mathbf b_j<0$ corresponds to a positive, that is attractive, interaction force along the axis rod $\Gamma_j$. Since $\alpha\to0$ at spatial infinity, the semi-infinite rods satisfy:
\begin{equation}\label{eq:intro-outer-defects}
\mathbf b_1=\mathbf b_{N+1}=0.
\end{equation}

We refer to a map $\Phi$ satisfying \eqref{eq:intro-harmonic-map-system} with the prescribed rod data as a \emph{stationary vacuum singular harmonic map} if it satisfies: the asymptotic behavior at regular axis points and nondegenerate poles established in \cite{LiTian1991,LiTian1992,LiTian1993,Weinstein1990,Weinstein1992,Weinstein1995}, the standard smooth horizon boundary conditions on the interiors of nondegenerate horizons, and the puncture and spatial infinity asymptotics formulated in \cite{HKWXAsymptotic}. The precise analytic estimates used in our proof are summarized in Proposition~\ref{prop:inputs}. The metric reconstructed from such a map is Ricci flat on the domain of outer communications, but is allowed \emph{a priori} to carry conical singularities along the bounded axis rods. The balance problem consists in deciding whether parameters can be chosen so that these singularities vanish simultaneously.

\begin{theorem}\label{thm:main}
Let $N\geq 2$. Consider an ordered collection $\mathcal H_1,\ldots,\mathcal H_N$ of horizon components, and let
\[
\Phi=(u,v):\R^3\setminus\Gamma\longrightarrow\Hh^2
\]
be a stationary vacuum singular harmonic map with this rod data. Let $w$, $\alpha$, and the Weyl--Papapetrou metric $\mathbf g$ be reconstructed from $\Phi$ via \eqref{eq:intro-WP-metric}--\eqref{eq:intro-alpha-quadrature} subject to the asymptotically flat normalization. Assume that
\[
\cJ_i\neq0
\]
for every degenerate horizon component. Then every bounded axis rod carries a strictly negative logarithmic angle defect:
\begin{equation}\label{eq:main-defect}
\mathbf b_j<0 \quad \text{for all } 2\leq j\leq N.
\end{equation}
\end{theorem}

In the special case where every horizon component is degenerate, Theorem~\ref{thm:main} yields the conjectured mass-angular momentum inequality $m\geq\sqrt{|\cJ|}$ for multiple dynamical black holes. Like the Penrose inequality, this relation was conjectured based on Penrose's cosmic censorship and collapse heuristics \cite{Penrose}; equality holds if and only if the data arise from a constant-time slice of an extreme Kerr spacetime. This result will be presented in detail in a forthcoming paper \cite{HKWXMass}.

Our proof is based on a PDE approach, and does not rely on the complete integrability of the equations. The argument hinges on two main components.
We first prove a global upper bound for $\alpha$,  using a new algebraic identity which yields the nonlinear differential inequality
\begin{equation}\label{eq:intro-key}
\Delta_{\rho,z}\alpha+\rho^{-1}|\grad\alpha| \geq 2e^{4u}|\grad v|^2 \geq 0,
\end{equation}
where $\Delta_{\rho,z} = \partial_\rho^2 + \partial_z^2$. Applying the maximum principle on an exhaustion of domains avoiding the axis, and controlling the boundary behavior at both horizon poles \cite{LiTian1992} and punctures \cite{HKWXAsymptotic}, we establish the global upper bound
\begin{equation}\label{eq:intro-global-bound}
\alpha \leq \max_{1\leq j\leq N+1}\mathbf b_j.
\end{equation}

We next prove that the maximum angle defect cannot be achieved along any bounded axis rod. Starting from the known $C^{3,\gamma}$ regularity across an open axis rod, the system forces the factorization
\begin{equation}\label{eq:intro-factorization}
U(\rho,z)=F(\rho^2,z),
\quad
v(\rho,z)=c+\rho^4 G(\rho^2,z),
\end{equation}
with $F$ and $G$ smooth. This is achieved by lifting $(v-c)/\rho^4$ radially to six transverse dimensions. Because the trace $U(0,z)$ tends to $-\infty$ at both endpoints of any bounded axis rod, it must attain an interior maximum along $\Gamma_j$. Combining \eqref{eq:intro-factorization} with Taylor expansions near this maximum, we deduce that $\alpha>\mathbf b_j$ at neighboring points unless the leading coefficients of both $U$ and $v$ vanish identically.  If $v$ is constant, the conclusion follows directly from harmonicity of $U$. When $v$ is non-constant, strong unique continuation together with the classification of rotationally symmetric homogeneous harmonic polynomials in dimensions three and seven identifies the leading non-vanishing powers. A careful analysis of $U$ and $v$ then shows that again
\begin{equation}\label{strictineq1}
\alpha > \mathbf b_j
\end{equation}
at points near the interior maximum point of $U$ on $\Gamma_j$.

Inequality \eqref{strictineq1} implies that $\max_k \mathbf b_k$ cannot be attained on a bounded axis rod. Since $\mathbf b_1 = \mathbf b_{N+1} = 0$, we conclude that every bounded rod must satisfy $\mathbf b_j < 0$, completing the proof.

The paper is organized as follows.
Section~\ref{sec:setup} recalls the relevant asymptotic analysis of the harmonic map equations.  The regular axis factorization \eqref{eq:intro-factorization} is obtained in Section \ref{sec:axis}, while Section~\ref{sec:nonlinear} is dedicated to the proof of the differential inequality \eqref{eq:intro-key}, and global upper bound \eqref{eq:intro-global-bound}. In Section~\ref{sec:strict} we establish the strict local inequality \eqref{strictineq1} near bounded axis rods, and Section~\ref{sec:proof-main} contains the proof of the main theorem.

\subsection*{Acknowledgements}
The authors thank Demetrios Christodoulou, who introduced the third author to this problem more than four decades ago. This project was supported in part by a SQuaRE workshop at the American Institute of Mathematics (AIM); the authors thank AIM for its hospitality and for providing a supportive and mathematically rich environment. The authors also thank Nan Wu for helpful discussions.

\subsection*{AI Disclosure} 
Various AI models were used in the development of this paper for computation, validation, and exploration.
In particular, one of the models suggested ideas which eventually led to identity \eqref{eq:alpha-laplacian} and inequality \eqref{eq:key-inequality}.

\section{The Singular Harmonic Map}\label{sec:setup}

In this section, we collect known results concerning the relevant asymptotics for the singular harmonic maps at the interior of axis rods, at poles and punctures, and at infinity. Recall for later use that if $h=h(\rho,z)$ is an axisymmetric function then
\[
 \Delta h=h_{\rho\rho}+\frac1\rho h_\rho+h_{zz},
 \quad
 \Delta_{\rho,z}h=h_{\rho\rho}+h_{zz},
 \quad
 |\grad h|^2=h_\rho^2+h_z^2,
\]
where $\Delta$ is the Euclidean Laplacian on $\R^3$ and $\Delta_{\rho,z}$ is the ordinary Laplacian on the orbit half-plane.
In particular, since $\Delta\ln\rho=0$ for $\rho>0$, the harmonic map equations \eqref{eq:intro-harmonic-map-system} can be expressed as
\begin{align}
 U_{\rho\rho}+\frac1\rho U_\rho+U_{zz}
 &=2e^{4u}(v_\rho^2+v_z^2),\label{eq:HM-U}\\
 v_{\rho\rho}+\frac1\rho v_\rho+v_{zz} &=
 -4\bigl((U_\rho-\rho^{-1})v_\rho+U_zv_z\bigr).
 \label{eq:HM-v-U}
\end{align}

\subsection{Set up and hypotheses}\label{hypotheses}
Let $\Phi=(u,v)\in H^1_{\mathrm{loc}}
 \bigl(\R^3\setminus\Gamma,\Hh^2\bigr)$
be an axisymmetric map satisfying the harmonic map equations \eqref{eq:intro-harmonic-map-system} weakly.
The energy and renormalized energy densities are, respectively,
\[
 e(u,v)=|\nabla u|^2+e^{4u}|\nabla v|^2,
 \quad
 e'(U,v)=|\nabla U|^2+e^{4u}|\nabla v|^2.
\]
Suppose that the $z$-axis $\Gamma$ has the finite mixed black hole rod
structure described in Section \ref{sec1}, with axis rods
$\Gamma_1,\ldots,\Gamma_{N+1}$ and horizon components
$\mathcal H_1,\ldots,\mathcal H_N$.  Let
$\mathcal I_{\mathrm{deg}}$ and $\mathcal I_{\mathrm{nd}}$ denote,
respectively, the sets of degenerate and nondegenerate horizon
components.  Consider the following three assumptions.

\begin{enumerate}[label=\textup{(A\arabic*)}]

\item\label{ass:local-regularity}
The appropriate local renormalizations are bounded and have finite
renormalized energy at ordinary axis points, at interiors of
nondegenerate horizon rods, and at their poles.  More precisely, the following conditions hold.

\begin{enumerate}[label=\textup{(\alph*)}]

\item
For every compact interval $I\subset\Gamma_j$ contained in an axis
rod, there is a neighborhood $\mathcal O_I \subset\mathbb{R}^3$ of $I$ such that
\begin{equation}\label{eq:axis-input-hypotheses}
 U,v\in L^\infty(\mathcal O_I\setminus\Gamma),
 \quad
 e'(U,v)\in L^1(\mathcal O_I),
\end{equation}
and
\[
 v=c_j\quad\text{on }I
\]
in the trace sense.

\item
For every $i\in\mathcal I_{\mathrm{nd}}$ and every compact interval
$I\subset\text{Int}\mathcal{H}_i$ contained in the interior of
the horizon rod, there is a neighborhood $\mathcal O_I \subset\mathbb{R}^3$ of
$I$ such that
\begin{equation}\label{eq:horizon-input-hypotheses}
 u,v\in L^\infty(\mathcal O_I\setminus\Gamma),
 \quad
 e(u,v)\in L^1(\mathcal O_I).
\end{equation}

\item
If $q_i^\pm=(0,0,z_i^\pm)$
are poles of a nondegenerate horizon rod, write
\[
 \zeta_i^\pm=z-z_i^\pm,
 \quad
 r_i^\pm=\sqrt{\rho^2+(\zeta_i^\pm)^2},
\]
and define the pole-renormalized functions
\begin{align}
 \widehat U_i^-
 &=U-\frac12\ln(r_i^--\zeta_i^-)
   =u+\frac12\ln(r_i^-+\zeta_i^-),
 \label{eq:lower-pole-renormalization}\\
 \widehat U_i^+
 &=U-\frac12\ln(r_i^++\zeta_i^+)
   =u+\frac12\ln(r_i^+-\zeta_i^+).
 \label{eq:upper-pole-renormalization}
\end{align}
In a neighborhood of $q_i^\pm$, the pair
$(\widehat U_i^\pm,v)$ is locally bounded and satisfies
\begin{equation}\label{eq:pole-input-hypotheses}
 |\nabla\widehat U_i^\pm|^2
 +e^{4u}|\nabla v|^2\in L^1.
\end{equation}

\end{enumerate}

\item\label{ass:puncture}
For every $i\in\mathcal I_{\mathrm{deg}}$, let
$p_i=(0,0,z_i)$
and introduce polar coordinates $(r_i,\theta_i,\phi)$ centered at
$p_i$.  There are
constants $\varepsilon_i,\Lambda_i>0$ such that
\begin{equation}\label{eq:puncture-input-hypotheses}
 |u+\ln\sin\theta_i|\leq\Lambda_i
 \text{ in }B_{\varepsilon_i}(p_i)\setminus\Gamma,
 \quad
 e'(U,v)\in
 L^1_{\mathrm{loc}}
 \bigl(B_{\varepsilon_i}(p_i)\setminus\{p_i\}\bigr).
\end{equation}
Moreover,
\[
 v=c_i\text{ on the axis rod below }p_i,
 \quad
 v=c_{i+1}\text{ on the axis rod above }p_i
\]
in the trace sense, where
\begin{equation}\label{eq:puncture-trace-jump}
 c_{i+1}-c_i=4\cJ_i\neq0.
\end{equation}

\item\label{ass:infinity}
There are constants $r_0,\Lambda_\infty>0$ such that
\begin{equation}\label{eq:infinity-input-hypotheses}
 |U|+|v|\leq\Lambda_\infty
 \text{ in }\overline B_{r_0}^{\,c}\setminus\Gamma,
 \quad
 e'(U,v)\in
 L^1_{\mathrm{loc}}\bigl(\overline B_{r_0}^{\,c}\bigr),
\end{equation}
and $v$ has constant traces $c_1$ and $c_{N+1}$ on the two
semi-infinite axis rods. Moreover, $U\rightarrow 0$ as $r\rightarrow\infty$.

\end{enumerate}

\subsection{Existence and uniqueness} 

\begin{proposition}
\label{prop:existence} Given a rod structure configuration of horizon
components $\{\mathcal{H}_i \}_{i=1}^{N}$, and axis rods $\{\Gamma_j \}_{j=1}^{N+1}$ with
potential constants $\{c_j\}_{j=1}^{N+1}$ such that $\mathcal{J}_i\neq 0$ for each $i\in \mathcal{I}_{\mathrm{deg}}$,
as in \eqref{eq:intro-horizon-components}, \eqref{eq:intro-rods}, 
\eqref{eq:intro-twist-potential-constant}, and \eqref{eq:intro-angular-momentum-jump}, respectively, 
there is a unique axisymmetric harmonic map $\Phi:\R^3\setminus\Gamma\rightarrow\Hh^2$
satisfying $(\mathrm{A}1)-(\mathrm{A}3)$. 
\end{proposition}

The proof of Proposition \ref{prop:existence} consists of two steps. In the first step, for the given $\{\mathcal{H}_i \}_{i=1}^{N}$, $\{\Gamma_j \}_{j=1}^{N+1}$, $\{c_j\}_{j=1}^{N+1}$, and $\mathcal{J}_i\neq 0$ as above, we prove that there exists a smooth axisymmetric model map $\Phi_0 :\R^3\setminus\Gamma\rightarrow\Hh^2$
realizing the prescribed rod structure and potential constants.
The model map is constructed by superposition of appropriate individual Kerr harmonic maps.
In the second step, we prove that there is a unique axisymmetric harmonic map $\Phi:\R^3\setminus\Gamma\rightarrow\Hh^2$ 
which is asymptotic to $\Phi_0$ in the sense that
\begin{equation}\label{eq:model-map-asymptotic-class}
 \sup_{\R^3\setminus\Gamma}
 d_{\Hh^2}(\Phi,\Phi_0)<\infty,
 \quad
 d_{\Hh^2}(\Phi,\Phi_0)\rightarrow0
 \quad\text{as }r\rightarrow\infty,
\end{equation}
and also satisfies $(\mathrm{A}1)-(\mathrm{A}3)$.
Refer to the survey \cite{Weinstein} for details. 

\subsection{Regularity and asymptotics}

\begin{proposition}
\label{prop:inputs}
Let $\Phi=(u,v):\mathbb{R}^3 \setminus\Gamma\rightarrow\mathbb{H}^2$ be a harmonic map satisfying $(\mathrm{A}1)-(\mathrm{A}3)$. 
Then the following statements hold.

\begin{enumerate}[label=\textup{(\roman*)}]

\item\label{concl:open-axis}
\emph{Axis regularity.}
On every compact subinterval of an axis rod $\Gamma_j$, the pair
$(U,v)$ extends across the axis as an axisymmetric
$C^{3,\gamma}$ pair for every $\gamma\in(0,1)$.  

\item\label{concl:puncture}
\emph{Degenerate horizon punctures.}
Let $i\in\mathcal I_{\mathrm{deg}}$. Set
$a_i=2|\cJ_i|$ and $\epsilon_i=\sgn\cJ_i$. 
Then there exists $b_i\in(-1,1)$ such that
\begin{equation}\label{eq:puncture-decomposition}
 U=\ln r_i+\bar U_i(\theta_i)+\widetilde U_i,
 \quad
 v=\bar v_i(\theta_i)+\widetilde v_i,
\end{equation}
where
\begin{align}
 \bar U_i(\theta_i)
 &=-\frac12\ln\left(
 \frac{2a_i\sqrt{1-b_i^2}}{1+\cos^2 \theta_i +2b_i \cos\theta_i}
 \right),
 \label{eq:bar-U-i}\\
 \bar v_i(\theta_i)
 &=\epsilon_i a_i
 \left(\frac{b_i+b_i\cos^2 \theta_i+2\cos\theta_i}{1+\cos^2 \theta_i +2b_i \cos\theta_i}\right)
 +v_i^0,
 \quad
 v_i^0=\frac{c_i+c_{i+1}}2.
 \label{eq:bar-v-i}
\end{align}
Moreover, if
\[
 \bar u_i=\bar U_i-\ln\sin\theta_i,
\]
then for every fixed $\varsigma\in(0,1)$, there are
$\beta_i\in(0,1)$ and $C_i>0$ such that
\begin{equation}\label{eq:puncture-weighted-estimate}
 \max_{\mathbb S^2}
 \left(
 \left|
 (r_i\partial_{r_i})^\ell
 \nabla_{\mathbb S^2}^{\,k}\widetilde U_i
 \right|
 +
 e^{(3+\varsigma-k)\bar u_i}
 \left|
 (r_i\partial_{r_i})^\ell
 \nabla_{\mathbb S^2}^{\,k}\widetilde v_i
 \right|
 \right)
 \leq C_i r_i^{\beta_i},
 \qquad
 \ell+k\leq3.
\end{equation}
In addition, 
\begin{equation}\label{eq:alpha-tangent-error}
 |\alpha-\bar\alpha_i|
 +r_i\left|
 \partial_{r_i}(\alpha-\bar\alpha_i)
 \right|
 +\left|
 \partial_{\theta_i}(\alpha-\bar\alpha_i)
 \right|
 \leq C_i r_i^{\beta_i}
\end{equation}
uniformly for $0\leq\theta_i\leq\pi$, where $\bar\alpha_i$ is given by 
\begin{equation}\label{eq:puncture-alpha-profile}
 \bar\alpha_i(\theta_i)
 =
 \frac{\mathbf b_i+\mathbf b_{i+1}}2
 +\ln\left(
 \frac{1+\cos^2 \theta_i +2b_i \cos\theta_i}{2\sqrt{1-b_i^2}}
 \right).
\end{equation}
Furthermore,
\begin{equation}\label{eq:defect-jump}
 \mathbf b_{i+1}-\mathbf b_i
 =\ln\left(\frac{1+b_i}{1-b_i}\right)
 =2\operatorname{arctanh}b_i.
\end{equation}

\item\label{concl:nondegenerate}
\emph{Nondegenerate horizon poles and interiors.}
Let $i\in\mathcal I_{\mathrm{nd}}$.  There are constants
$\kappa_i^\pm\in\R$ and exponents
$\gamma_i^\pm\in(0,1)$ such that in a neighborhood of the poles
\begin{align}
 U
 =\frac12\ln(r_i^--\zeta_i^-)
   +\kappa_i^-+R_i^-,\quad
 |R_i^-|+r_i^-|\nabla R_i^-|
 \leq C(r_i^-)^{\gamma_i^-},
 \label{eq:lower-pole-U}
\end{align}
and  
\begin{align}
 U
 =\frac12\ln(r_i^++\zeta_i^+)
   +\kappa_i^++R_i^+,
 \quad
 |R_i^+|+r_i^+|\nabla R_i^+|
 \leq C(r_i^+)^{\gamma_i^+},
 \label{eq:upper-pole-U}
\end{align}
while the corresponding Weyl conformal factor $\alpha$ expansions are
\begin{align}
 \alpha(r_i^-,\theta_i^-)
 &=\mathbf b_i+\ln\sin\frac{\theta_i^-}{2}
   +O\bigl((r_i^-)^{\gamma_i^-}\bigr),
 \label{eq:lower-pole-alpha}\\
 \alpha(r_i^+,\theta_i^+)
 &=\mathbf b_{i+1}+\ln\cos\frac{\theta_i^+}{2}
   +O\bigl((r_i^+)^{\gamma_i^+}\bigr).
 \label{eq:upper-pole-alpha}
\end{align}
Furthermore, on every compact subinterval $I\subset\mathrm{Int}\mathcal{H}_i$,
the functions $(u,v)$ extend across the rod
as an axisymmetric smooth pair, and there is a constant $\kappa_i\in\R$, depending only
on the horizon rod $\mathcal H_i$, such that in a neighborhood of $I$,
\begin{equation}\label{eq:horizon-alpha-refined}
 \alpha(\rho,z)
 =\ln\rho+2u(\rho,z)+\kappa_i+O(\rho^2)
 \quad\text{as }\rho\rightarrow 0.
\end{equation}

\item\label{concl:infinity}
\emph{Infinity expansion.}
Set
\[
 \cJ=\frac{c_{N+1}-c_1}{4},
 \quad
 c_\infty=\frac{c_{N+1}+c_1}{2}.
\]
For every fixed $\varsigma\in(0,1)$ there are constants
$c_1^\infty\in\R$ and a remainder $\widetilde v_\infty$ such that
\begin{align}
 U=\frac{c_1^\infty}{r}
   +O_{1}(r^{-2}),
 \quad
 v=\cJ\cos\theta(3-\cos^2\theta)+c_\infty
   +\widetilde v_\infty,
 \label{eq:inf-v}
\end{align}
where
\begin{equation}\label{eq:inf-v-der}
 |\partial_r\widetilde v_\infty|
 \leq Cr^{-2}\sin^{3+\varsigma}\theta,
 \quad
 |\partial_\theta\widetilde v_\infty|
 \leq Cr^{-1}\sin^{2+\varsigma}\theta.
\end{equation}
Moreover, the Weyl conformal factor $\alpha$ satisfies
\begin{equation}\label{eq:inf-alpha}
 \alpha=O(r^{-2}),
 \quad
 \partial_r\alpha=O(r^{-3}),
 \quad
 \partial_\theta\alpha=O(r^{-2}).
\end{equation}
\end{enumerate}
\end{proposition}

The puncture and infinity statements are from \cite[Theorems~2.1--2.3 and equations (7.2)--(7.4)]{HKWXAsymptotic}, while the axis and pole statements are based on Li--Tian and Weinstein \cite{LiTian1992,LiTian1993,Weinstein1990,Weinstein1992,Weinstein1995}; see also
the summary in \cite[Section~2, discussion following Theorem~2.3]{HKWXAsymptotic}. The corresponding statements concerning the Weyl conformal factor $\alpha$ arise directly from the harmonic map asymptotics and equations \eqref{eq:intro-alpha-quadrature}. 

\section{Regular Axis Factorization}\label{sec:axis}

In this section, the factorization \eqref{eq:intro-factorization} will be established. Results in this section will be needed in Sections \ref{sec:nonlinear} and \ref{sec:strict}. 
We first present some elementary facts about radial functions. 


\begin{lemma}\label{lem:radial-profiles}
Let $I\subset\R$ be an interval and let $m,\ell\geq2$.
\begin{enumerate}[label=\textup{(\roman*)}]
\item Suppose that $q=q(\rho,z)$ is defined for $\rho\geq0$, and that
\[
 h_m(x,z)=q(|x|,z),\quad (x,z)\in\R^m\times I,
\]
is $C^{k,\gamma}$ and $\SO(m)$-invariant, where $k\geq0$ and $0<\gamma<1$.  Then the same radial representative defines an $\SO(\ell)$-invariant function
\[
 h_\ell(X,z)=q(|X|,z)
\]
of class $C^{k,\gamma}$ on $\R^\ell\times I$.

\item If $h_m$ is smooth, then there is a function $H=H(s,z)$, which is smooth for $s\geq0$ up to $s=0$, such that
\begin{equation}\label{eq:smooth-radial-factorization}
 h_m(x,z)=H(|x|^2,z).
\end{equation}

\item If $q(|x|,z)$ is $C^{k+2,\gamma}$ on $\mathbb{R}^m \times I$, then $q_\rho /\rho$, initially defined for $\rho>0$, extends as a radial $C^{k,\gamma}$ function across $x=0$.  
\end{enumerate}
\end{lemma}

\begin{proof}
The proof is straightforward and included for completeness.
We may work locally on a compact subinterval $I'\subset I$. Let $e_1 \in \mathbb{R}^m$ be a unit vector. Restriction to the line defined by $e_1$ gives an even extension
\[
 Q(\rho,z):=h_m(\rho e_1,z),\quad (\rho,z)\in(-\rho_0,\rho_0)\times I',
\]
of class $C^{k,\gamma}$.  At a point of the $z$-axis, every Taylor polynomial of an $\SO(m)$-invariant function is an invariant polynomial,  and hence is a linear combination of the powers $|x|^{2j}$.  Applying Taylor's theorem in the transverse variables $x\in \mathbb{R}^m$, also after $n$ differentiations in $z$ with $n\leq k$, gives
\begin{equation}\label{eq:radial-finite-taylor}
 \partial_z^n Q(\rho,z)
 =\sum_{2j\leq k-n}a_{j,n}(z)\rho^{2j}
 +R_{k-n,n}(\rho,z),
\end{equation}
where, for $0\leq \ell\leq k-n$,
\begin{equation}\label{eq:radial-remainder-estimate}
 \bigl|\partial_\rho^\ell R_{k-n,n}(\rho,z)\bigr|
 \leq C|\rho|^{k-n+\gamma-\ell},
\end{equation}
and the corresponding H\"older estimates hold uniformly for $z\in I'$.  In particular, all odd radial derivates permitted by the regularity vanish at $\rho=0$.

For $X\neq0$, repeated differentiation of a radial representative has the form
\begin{equation}\label{eq:radial-cartesian-derivatives}
 \partial_X^\beta\bigl[Q(|X|,z)\bigr]
 =\sum_{\ell=1}^{|\beta|}
 \partial_\rho^\ell Q(|X|,z)|X|^{\ell-|\beta|}
 P_{\beta,\ell}\!\left(\frac{X}{|X|}\right),
\end{equation}
where the $P_{\beta,\ell}$ are smooth functions on the unit sphere depending only on $\beta$ and $\ell$.
By substituting \eqref{eq:radial-finite-taylor} into \eqref{eq:radial-cartesian-derivatives}, we find that the polynomial part 
is a polynomial in $X$, while \eqref{eq:radial-remainder-estimate} controls every apparent negative power of $|X|$ and gives the required $C^{k-n,\gamma}$ extension at $X=0$.  Applying this for each $n\leq k$ proves that $Q(|X|,z)$ is $C^{k,\gamma}$, jointly in $(X,z)$ .  This does not depend on the transverse dimension, and thus proves (i).

For (ii), observe that the evenness of $Q$ and repeated Taylor expansion give, for every $L$,
\[
 Q(\rho,z)=\sum_{j=0}^{L}a_j(z)\rho^{2j}+\rho^{2L+2}Q_{L+1}(\rho,z),
\]
where $Q_{L+1}$ is smooth and even.  Hence, $H(s,z):=Q(\sqrt{s},z)$ for $s\geq0$ has derivatives of every order extending continuously to $s=0$.  Rotational invariance of $h_m$ then yields \eqref{eq:smooth-radial-factorization}.

Lastly, evenness implies $q_\rho(0,z)=0$, and therefore
\begin{equation}\label{eq:radial-quotient-formula}
 \frac{q_\rho(\rho,z)}{\rho}=\int_0^1 q_{\rho\rho}(t \rho,z)dt.
\end{equation}
Moreover, since $Q(\rho,z)=q(|\rho|,z)$ is even and $C^{k+2,\gamma}$, its second radial
derivative $Q_{\rho\rho}$ is even and $C^{k,\gamma}$.  The finite regularity
argument from part (i), applied to $Q_{\rho\rho}$, shows that the radial
function $(x,z)\mapsto q_{\rho\rho}(|x|,z)$
is $C^{k,\gamma}$. Therefore, $|x|^{-1}q_\rho(|x|,z)$
extends as a $C^{k,\gamma}$ function across
$x=0$.
This proves (iii).
\end{proof}

The next result gives the desired smooth factorization of the harmonic map functions at the interior of axis rods.

\begin{proposition}\label{prop:axis-factorization}
Let $I$ be a compact subinterval of an axis rod $\Gamma_j$, and write $v=c$ on that rod. There exist functions $F$ and $G$, smooth on $\mathbb{R}^2_+$ up to the boundary, such that on a sufficiently small neighborhood of $I$ in $\mathbb{R}^3$,
\begin{equation}\label{eq:axis-factorization}
 U(\rho,z)=F(\rho^2,z),
 \quad
 v(\rho,z)=c+\rho^4 G(\rho^2,z).
\end{equation}
\end{proposition}

\begin{proof} 
Since $u=U-\ln\rho$, equation \eqref{eq:HM-v-U} becomes
\begin{equation}\label{eq:y-equation}
 v_{\rho\rho}-\frac3\rho v_\rho+v_{zz}
 +4U_\rho v_\rho+4U_zv_z=0.
\end{equation}
By Proposition~\ref{prop:inputs}(i), $v$ is the restriction of an $\SO(2)$-invariant Cartesian $C^{3,\gamma}$ function in $(x_1,x_2,z)$, and $v(0,z)=c$.  
Uniformly on compact subintervals of the rod,
\begin{align}
 v(\rho,z)&=c+a(z)\rho^2+O(\rho^{3+\gamma}),\notag\\
 v_\rho(\rho,z)&=2a(z)\rho+O(\rho^{2+\gamma}),\label{eq:y-taylor}\\
 v_{\rho\rho}(\rho,z)&=2a(z)+O(\rho^{1+\gamma}),\notag
\end{align}
with the analogous estimates involving the $z$-differentiations allowed by $C^{3,\gamma}$ regularity.  Since $v(0,z)$ is identically constant, $v_{zz}(0,z)=0$.  Moreover,
\[
 U_\rho=O(\rho),\quad U_z=O(1),\quad v_z=O(\rho^2).
\]
Letting $\rho\rightarrow0$ in \eqref{eq:y-equation} and using \eqref{eq:y-taylor} produces
\[
 2a(z)-6a(z)=0.
\]
Thus $a(z)\equiv0$ on the rod interval, and we have
\begin{equation}\label{eq:y-improved}
 |v-c|\leq C\rho^{3+\gamma},
 \quad |v_\rho|+|v_z|\leq C\rho^{2+\gamma}.
\end{equation}

For $\rho>0$, define
\begin{equation}\label{eq:def-psi}
 h(\rho,z)=\rho^{-4}\big[v(\rho,z)-c\big].
\end{equation}
A direct substitution into \eqref{eq:y-equation} gives
\begin{equation}\label{eq:psi-six-radial}
 h_{\rho\rho}+\frac5\rho h_\rho+h_{zz}
 +4U_\rho h_\rho+4U_z h_z
 +16\frac{U_\rho}{\rho} h=0.
\end{equation}
The estimates \eqref{eq:y-improved} imply
\begin{equation}\label{eq:psi-initial-growth}
 |h|\leq C\rho^{-1+\gamma},
 \quad |\grad h|\leq C\rho^{-2+\gamma}.
\end{equation}
Now lift $U$ and $h$ to six transverse variables $X\in\R^6$ by
\[
\widehat U(X,z)=U(|X|,z),
\quad
 \widehat G(X,z)=h(|X|,z).
\]
Away from $X=0$, $\widehat G$ satisfies 
\begin{equation}\label{eq:psi-lifted}
 \Delta_{\R^7}\widehat G
 +4\grad\widehat U\cdot\grad\widehat G
 +16\widehat{\left(\frac{U_\rho}{\rho}\right)}\widehat G=0.
\end{equation}
We next improve the regularity of $\widehat G$ and prove that \eqref{eq:psi-lifted} holds across $X=0$. 

By Lemma~\ref{lem:radial-profiles}(i), $\widehat U$ is $C^{3,\gamma}$, and Lemma~\ref{lem:radial-profiles}(iii) shows that $U_\rho/\rho$ extends at least as a $C^{1,\gamma}$ function across $X=0$. Moreover from \eqref{eq:psi-initial-growth}, in a bounded $z$-cylinder, we have
\begin{align}\label{eq:psi-L2}
 \int_0^\eps |h|^2\rho^5 d\rho=O(\eps^{4+2\gamma}),\qquad
 \int_0^\eps |\grad h|^2\rho^5 d\rho=O(\eps^{2+2\gamma}).
\end{align}
Thus $\widehat G\in H^1_{\mathrm{loc}}(\R^7)$.
Choose a radial cut-off $\chi_\eps(X)$ equal to one for $|X|<\eps$, zero for $|X|>2\eps$, and satisfying $|\grad\chi_\eps|\leq C\eps^{-1}$.  Then
\begin{equation}\label{eq:capacity-cutoff}
 \int_{\R^6}|\grad\chi_\eps|^2 dX
 \leq C\eps^{-2}\int_\eps^{2\eps}\rho^5 d\rho
 =O(\eps^4).
\end{equation}
Use $(1-\chi_\eps)\varphi$ as a test function for \eqref{eq:psi-lifted}, where $\varphi\in C_c^\infty(\R^7)$.  The cut-off term involving the principal part is bounded by
\begin{align*}
 \left(\int_{\{\eps<|X|<2\eps\}}|\grad\widehat G|^2\right)^{1/2}
 \left(\int_{\{\eps<|X|<2\eps\}}|\grad\chi_\eps|^2\right)^{1/2}\|\varphi\|_{L^\infty}
 =O(\eps^{1+\gamma})\cdot O(\eps^2)=o(1),
\end{align*}
by \eqref{eq:psi-L2} and \eqref{eq:capacity-cutoff}.  Since the square root of the tube volume is $O(\eps^3)$, the remaining omitted error terms satisfy
\[
 O(\eps^{1+\gamma})\cdot O(\eps^3),\quad
 O(\eps^{4+\gamma}),\quad
 O(\eps^{5+\gamma}),
\]
for the principal term containing $\grad\varphi$, the first-order term, and the zeroth-order term, respectively.  Here we used \eqref{eq:psi-L2}, and boundedness of the coefficients. It follows that \eqref{eq:psi-lifted} holds weakly across $\{X=0\}$.

We may now bootstrap.  The lower-order coefficients in \eqref{eq:psi-lifted} are bounded, and $\widehat G\in H^1$, so $ \Delta_{\R^7}\widehat G\in L^2$.  Interior $W^{2,p}$ estimates and Sobolev embedding successively raise the integrability of $\grad\widehat G$ according to $2\rightarrow\frac{14}{5}\rightarrow\frac{14}{3}\rightarrow 14$, since in 
dimension seven the critical Sobolev exponent is $7p/(7-p)$.  After the last step, $\widehat G\in W^{2,14}_{\mathrm{loc}}$, and Morrey's inequality yields $\widehat G\in C^{1,1/2}_{\mathrm{loc}}$.  If $0<\gamma_0<\min\{\gamma,\tfrac12\}$, then
Schauder estimates applied to \eqref{eq:psi-lifted} produce $\widehat G\in C^{3,\gamma_0}$.

Next observe that the $U$ equation \eqref{eq:HM-U}, expressed in terms of $h$, is given by
\begin{equation}\label{eq:U-psi-equation}
 U_{\rho\rho}+\frac1\rho U_\rho+U_{zz}
 =2e^{4U}\left[\rho^2(4h+\rho h_\rho)^2
 +\rho^4 h_z^2\right].
\end{equation}
Lemma \ref{lem:radial-profiles} allows us to move the improved radial representative of $\widehat G$ back to the original two transverse variables.  
Since the right-hand side of \eqref{eq:U-psi-equation} is then a $C^{2,\gamma_0}$ function in $\R^3$ across the original axis, the three-dimensional Schauder estimates show that $U\in C^{4,\gamma_0}$.

We now proceed inductively.  Starting from $m=0$, suppose that
\[
 U\in C^{m+3,\gamma_0},
 \quad \widehat G\in C^{m+2,\gamma_0}.
\]
By Lemma~\ref{lem:radial-profiles}(iii), it holds that $U_\rho/\rho\in C^{m+1,\gamma_0}$.  Equation \eqref{eq:psi-lifted} then yields $\widehat G\in C^{m+3,\gamma_0}$, and thus \eqref{eq:U-psi-equation} gives $U\in C^{m+4,\gamma_0}$. Continuing this iteration shows that both $U$ and the lifted $\widehat G$ are smooth across the axis.
We may now apply Lemma~\ref{lem:radial-profiles}(ii), first in two transverse dimensions for $U$ and then in six transverse dimensions for $\widehat G$, to find smooth functions $F$ and $G$ with
\[
 U(\rho,z)=F(\rho^2,z),
 \quad h(\rho,z)=G(\rho^2,z).
\]
Although this holds only in a neighborhood of the interval $I$, the functions $F$ and $G$ may be extended smoothly to all of $\mathbb{R}^2_+$. 
Lastly, since $v-c=\rho^4h$, the desired result follows.
\end{proof}

\section{Weyl Conformal Factor Global Upper Bound}\label{sec:nonlinear}

The goal of this section is to obtain the global upper bound \eqref{eq:intro-global-bound} for $\alpha$. We begin by deriving the nonlinear
differential inequality \eqref{eq:intro-key}. Although the equation for $\Delta_{\rho,z}\alpha$ obtained by differentiating the quadratures \eqref{eq:intro-alpha-quadrature} has an indefinite right-hand side, the negative term is controlled exactly by $|\grad\alpha|/\rho$.  Moreover, our proof relates this inequality to the scalar curvature of constant time slices in the stationary vacuum spacetime.

\subsection{The Weyl conformal factor identity/inequality}

\begin{lemma}
\label{lem:alpha-laplacian-geometric}
Away from $\Gamma$, the Weyl conformal factor $\alpha$ satisfies the identity
\begin{equation}\label{eq:alpha-laplacian}
 \Delta_{\rho,z}\alpha
 =-|\nabla U|^2+3e^{4u}|\nabla v|^2,
\end{equation}
and consequently the differential inequality
\begin{equation}\label{eq:key-inequality}
 \Delta_{\rho,z}\alpha+\frac{|\grad\alpha|}{\rho}
 \geq 2e^{4u}|\grad v|^2\geq0.
\end{equation}
\end{lemma}

\begin{proof}
Consider a time slice $\{t=\mathrm{constant}\}$, of the stationary vacuum spacetime with metric \eqref{eq:intro-WP-metric}.  The metric induced on
$\Sigma_t$ is given by
\[
 h=e^{-2U}\bar q,
 \quad
 \bar q=e^{2\alpha}(d\rho^2+dz^2)+\rho^2d\phi^2.
\]
A standard Riemannian geometry formula yields the scalar curvature of this slice
\begin{equation}\label{eq:scalar-conformal-alpha}
 R_h
 =e^{2U}\left(
 \bar R+4\bar\Delta U-2|\bar\nabla U|_{\bar q}^2
 \right).
\end{equation}
Moreover, for the metric $\bar q$, a direct calculation yields
\[
 \bar R=-2e^{-2\alpha}\Delta_{\rho,z}\alpha,
 \quad
 \bar\Delta U=e^{-2\alpha}\Delta U,
 \quad
 |\bar\nabla U|_{\bar q}^2=e^{-2\alpha}|\nabla U|^2,
\]
and hence
\begin{equation}\label{eq:scalar-h-geometric}
 R_h
 =e^{2U-2\alpha}
 \left(
 -2\Delta_{\rho,z}\alpha
 +4\Delta U
 -2|\nabla U|^2
 \right).
\end{equation}
On the other hand, we may compute the scalar curvature from the vacuum Hamiltonian constraint \cite[(2.17)]{KhuriWeinstein2016} to find
\begin{equation}\label{eq:scalar-h-constraint}
 R_h
 =2e^{2U-2\alpha}e^{4u}|\nabla v|^2.
\end{equation}
Furthermore, by the harmonic map equations
\begin{equation}\label{OL}
 \Delta U
 =\Delta u+\Delta \ln\rho
 =2e^{4u}|\nabla v|^2,
\end{equation}
since $\ln\rho=0$ is harmonic.  The desired identity \eqref{eq:alpha-laplacian} now follows by combining
\eqref{eq:scalar-h-geometric}-\eqref{OL}.

To obtain \eqref{eq:key-inequality}, first observe that 
the two equations in \eqref{eq:intro-alpha-quadrature} give
\begin{equation*}
 \frac{|\grad\alpha|^2}{\rho^2}
 =\left[U_\rho^2-U_z^2+e^{4u}(v_\rho^2-v_z^2)\right]^2
 +4\left[U_\rho U_z+e^{4u}v_\rho v_z\right]^2.
\end{equation*}
The elementary identity
\[
 \bigl[(X_1^2-X_2^2)+(Y_1^2-Y_2^2)\bigr]^2
 +4(X_1X_2+Y_1Y_2)^2
 =(|X|^2-|Y|^2)^2+4(X\cdot Y)^2
\]
with $X=\grad U$ and $Y=e^{2u}\grad v$ produces
\begin{equation*}
\rho^{-1}|\nabla \alpha|\geq \Big| |\nabla U|^2 -e^{4u}|\nabla v|^2 \Big|.
\end{equation*}
Combining this with \eqref{eq:alpha-laplacian} yields
\[
 \Delta_{\rho,z}\alpha+\rho^{-1}|\grad\alpha|
 \geq -|\nabla U|^2+3e^{4u}|\nabla v|^2+\Big| |\nabla U|^2 -e^{4u}|\nabla v|^2 \Big|
 \geq 2e^{4u}|\nabla v|^2,
\]
which proves \eqref{eq:key-inequality}.
\end{proof}

\subsection{Puncture, pole, and axis estimates}\label{sec:puncture}
We will use the asymptotics obtained in Section \ref{sec:setup} to provide estimates for the Weyl conformal factor
near punctures, poles, and axis, in preparation for the maximal principle argument.

\begin{lemma}\label{lem:tangent-alpha}
The Weyl conformal factor $\alpha$ satisfies the following estimates. 
\begin{enumerate}[label=\textup{(\roman*)}]
\item\emph{Punctures.} Let $i\in\mathcal I_{\mathrm{deg}}$, and let $(r_i,\theta_i,\phi)$ be polar coordinates
centered at puncture $p_i$. Then there exists a constant $C>0$ such that in a neighborhood of this puncture
\begin{equation}\label{eq:tangent-endpoint-bound}
 \sup_{0\leq\theta_i \leq\pi}\alpha(r_i,\theta_i)
 \leq\max\{\mathbf b_i,\mathbf b_{i+1}\}+Cr_i^{\beta_i}.
\end{equation}

\item\emph{Poles.} Let $i\in\mathcal I_{\mathrm{nd}}$, and let $(r_i^{\pm},\theta_i^{\pm},\phi)$ be polar coordinates
centered at poles $q_i^{\pm}$. Then there exists a constant $C>0$ such that in a neighborhood of these poles
\begin{equation}\label{eq:lower-pole-sup}
 \sup_{0<\theta_i^-\leq\pi}\alpha(r_i^{-},\theta_i^-)
 \leq \mathbf b_i+C(r_{i}^{-})^{\gamma_i^-},
\quad
 \sup_{0\leq\theta_i^+<\pi}\alpha(r_i^+,\theta_i^+)
 \leq \mathbf b_{i+1}+C(r_i^+)^{\gamma_i^+}.
\end{equation}

\item\emph{Axis.} Let $I$ be a compact subinterval of an axis rod $\Gamma_j$.  Then there exist constants $C_I, \rho_I>0$ such that
\begin{equation}\label{eq:rod-tube}
 \sup_{z\in I}|\alpha(\rho,z)-\mathbf b_j|
 \leq C_I\rho^2,
 \quad 0\leq\rho\leq\rho_I.
\end{equation}

\end{enumerate}
\end{lemma}

\begin{proof}
Consider the three regions separately.\medskip

\noindent (i) \emph{Punctures.} According to Proposition \ref{prop:inputs}(ii), in a neighborhood of the puncture
\begin{equation}\label{__}
\alpha(r_i,\theta_i)
 =
 \frac{\mathbf b_i+\mathbf b_{i+1}}2
 +\ln\left(
 \frac{1+\cos^2 \theta_i +2b_i \cos\theta_i}{2\sqrt{1-b_i^2}}
 \right)+O(r_i^{\beta_i}).
\end{equation}
Set $x=\cos\theta_i$, and note that the function $x \mapsto 1+x^2+2b_i x$ is convex, so that 
its maximum on $[-1,1]$ occurs at an endpoint.  Since $\ln$ is increasing, the second term on the right-hand side of \eqref{__} is bounded above by the larger of its two pole values $\theta_i=0,\pi$.  This, combined with \eqref{eq:defect-jump} yields the desired inequality \eqref{eq:tangent-endpoint-bound}.
\medskip

\noindent (ii) \emph{Poles.} 
Equations \eqref{eq:lower-pole-alpha} and \eqref{eq:upper-pole-alpha}, together with the elementary inequalities
$\ln\sin\frac{\theta_i^-}{2}\leq0$ and $\ln\cos\frac{\theta_i^+}{2}\leq0$
give the two desired estimates immediately.  
\medskip

\noindent (iii) \emph{Axis.} 
Proposition~\ref{prop:axis-factorization} gives, uniformly for $z\in I$,
\[
 U_\rho=O(\rho),
 \quad U_z=O(1),
 \quad v_\rho=O(\rho^3),
 \quad v_z=O(\rho^4).
\]
Since $e^{4u}=\rho^{-4}e^{4U}$, the two equations in  \eqref{eq:intro-alpha-quadrature} imply
\begin{equation}\label{eq:alpha-axis-orders}
 \alpha_\rho=O(\rho),
 \quad \alpha_z=O(\rho^2).
\end{equation}
Thus, since $\alpha(0,z)=\mathbf{b}_j$ we have
\[
 |\alpha(\rho,z)-\mathbf b_j|
 \leq\int_0^\rho|\alpha_\tau(\tau,z)|dt
 \leq C_I\rho^2,
\]
which yields the desired result.
\end{proof}

\subsection{Maximum Principle}
\label{sec:global}
We will now establish the global upper bound for $\alpha$, by applying the maximum principle on a sequence of exhaustion domains.

\begin{proposition}\label{prop:global-bound}
The Weyl conformal factor $\alpha$ satisfies the following estimate
\begin{equation}\label{eq:global-alpha-bound}
 \alpha(\rho,z)\leq \max_{1\leq j\leq N+1}\mathbf b_j
 \quad\text{for all }\rho>0\text { and } z.
\end{equation}
\end{proposition}

\begin{proof}
Set $M=\max_{1\leq j\leq N+1}\mathbf b_j$, and observe that $M\geq 0$, since according to
\eqref{eq:intro-outer-defects} the logarithmic angle defects of the semi-infinite rods vanish.
Let
\[
 \mathcal Q=
 \{p_i \mid i\in\mathcal I_{\mathrm{deg}}\}
 \cup\{q_i^-,q_i^+ \mid i\in\mathcal I_{\mathrm{nd}}\}
\]
be the finite set of punctures and poles, and for $q=(0,0,z_q)\in\mathcal Q$, put
\[
 r_q=\sqrt{\rho^2+(z-z_q)^2}.
\]
Choose $r_*>0$ so that all puncture and pole estimates from Lemma \ref{lem:tangent-alpha} hold for $r_q\leq2r_*$.  Since $\mathcal Q$ is finite, there is a common exponent
\[
 \kappa_*=
 \min\bigl(
 \{\beta_i \mid i\in\mathcal I_{\mathrm{deg}}\}
 \cup
 \{\gamma_i^-,\gamma_i^+ \mid i\in\mathcal I_{\mathrm{nd}}\}
 \bigr)>0.
\]
Fix $\sigma>0$ smaller than $r_*$ and smaller than a quarter of the distance between any two distinct points of $\mathcal Q$.  Choose $R$ so large that all points of $\mathcal{Q}$ lie in the half-disk of radius $R/2$, and take $0<\eta<\sigma/2$. 
Define
\begin{equation}\label{eq:exhaustion-domain}
 \Omega_{\eta,\sigma,R}
 =\left\{(\rho,z)\mid 
 \rho>\eta,\ \rho^2+z^2<R^2,
 \ r_q>\sigma\ \text{for every }q\in\mathcal Q
 \right\},
\end{equation}
and observe that this domain is bounded with piecewise smooth boundary.

On $\Omega_{\eta,\sigma,R}$ set
\begin{equation}\label{eq:drift-B}
 B(\rho,z)=
 \begin{cases}
 \displaystyle\frac{\grad\alpha}{\rho|\grad\alpha|},&\grad\alpha\neq0,\\[6pt]
 0,&\grad\alpha=0,
 \end{cases}
\end{equation}
and notice that $B$ is bounded and measurable, with $|B|\leq\eta^{-1}$. Furthermore, by Lemma~\ref{lem:alpha-laplacian-geometric} we have
\begin{equation}\label{eq:linearized-alpha-subsolution}
 \Delta_{\rho,z}\alpha+B\cdot\grad\alpha\geq0.
\end{equation}
The maximum principle then yields
\begin{equation}\label{eq:finite-max-principle}
 \sup_{\Omega_{\eta,\sigma,R}}\alpha
 \leq\sup_{\partial\Omega_{\eta,\sigma,R}}\alpha.
\end{equation}

We will now estimate the boundary pieces.  On a half-circle $r_q=\sigma$, Lemma~\ref{lem:tangent-alpha}(i)-(ii) apply when $q$ is a puncture or pole.  Since there are finitely many punctures or poles, we have
\begin{equation}\label{eq:marked-boundary-estimate}
 \alpha\leq M+C\sigma^{\kappa_*}.
\end{equation}
Next consider the boundary portions with $\rho=\eta$.  The condition $r_q>\sigma$ and the inequality $\eta<\sigma/2$ imply
\[
 |z-z_q|>\sqrt{\sigma^2-\eta^2}>\frac{\sqrt3}{2}\sigma.
\]
For fixed $\sigma$ and $R$, the projections of these portions onto $\Gamma$ lie in a finite union of compact subintervals, each contained either in an axis rod or in the interior of a nondegenerate horizon rod.  On the axis pieces, Lemma~\ref{lem:tangent-alpha}(iii) gives
\begin{equation}\label{eq:axis-boundary-estimate}
 \alpha\leq M+C_{\sigma,R}\eta^2,
\end{equation}
while on the horizon pieces, \eqref{eq:horizon-alpha-refined} gives
\begin{equation}\label{eq:horizon-boundary-estimate}
 \alpha\leq\ln\eta+C_{\sigma,R}.
\end{equation}
Lastly, on the outer half-circle \eqref{eq:inf-alpha} shows that
\begin{equation}\label{eq:outer-boundary-estimate}
 |\alpha|\leq CR^{-2}.
\end{equation}
Combining these estimates together with \eqref{eq:finite-max-principle} yields
\begin{equation}\label{eq:exhaustion-bound}
 \sup_{\Omega_{\eta,\sigma,R}}\alpha
 \leq\max\left\{
 M+C\sigma^{\kappa_*},\,
 M+C_{\sigma,R}\eta^2,\,
 \ln\eta+C_{\sigma,R},\,
 CR^{-2}
 \right\}.
\end{equation}

Fix $(\rho_0,z_0)$ arbitrarily with $\rho_0>0$.  For all sufficiently large $R$ and sufficiently small $\sigma,\eta$, this point belongs to $\Omega_{\eta,\sigma,R}$.  In \eqref{eq:exhaustion-bound}, first let $\eta\rightarrow 0$ with $\sigma,R$ fixed. Then let $\sigma\rightarrow 0$, and finally let $R\to\infty$.  Since $M\geq0$, 
it follows that $\alpha(\rho_0,z_0)\leq M$. 
\end{proof}

\section{Weyl Conformal Factor Strict Local Lower Bound}\label{sec:strict}

The purpose of this section is to prove a strict local lower bound for $\alpha$ near any bounded axis rod. The proof must cover arbitrarily degenerate maxima of the axis trace of $U$; this is the reason for the expansion by homogeneous harmonic polynomials below. We will use the first equation of \eqref{eq:intro-alpha-quadrature}, and demonstrate that, under careful analysis of integration over appropriate line segments, the $\rho$-derivative terms dominate the $z$-derivatives. This is reminiscent of the argument of Li--Tian \cite{LiTian1991}, where the $z$-derivatives vanish due to the assumed symmetry in that work.

\begin{proposition}\label{prop:strict-rod}
Let $\Gamma_i=(z_{i-1}^+,z_i^-)$, $2\leq i\leq N$, be a bounded open axis rod on which $v=c_i$ and $\alpha=\mathbf b_i$.  
If $q\in \Gamma_i$ is a maximum point of $z\mapsto U(0,z)$, then every neighborhood of $(0,q)$ in the 
half-plane $\rho>0$ contains a point where
\begin{equation}\label{eq:strict-lower-bound}
 \alpha>\mathbf b_i.
\end{equation}
\end{proposition}

\begin{proof}
After translating $q$ to the origin, we may assume that $q=0$.  Near the rod, by Proposition~\ref{prop:axis-factorization}, we can write
\begin{equation}\label{eq:strict-local-notation}
 U=F(s,z),
 \quad v=c_i+s^2G(s,z),
\end{equation}
where $s=\rho^2$ and $F$, $G$ are smooth.
Set
\[
 f(z)=F(0,z),
 \quad g(z)=G(0,z).
\]
Then $\alpha(0,z)=\mathbf b_i$ and $f'(0)=0$. 

We now record the relevant equations in terms of the functions $F$ and $G$, and with respect to the variable $s=\rho^2$. 
A direct substitution of \eqref{eq:strict-local-notation} into the harmonic map equations \eqref{eq:HM-U}-\eqref{eq:HM-v-U} and the first equation in  \eqref{eq:intro-alpha-quadrature} 
produces
\begin{align}
 4(sF_{ss}+F_s)+F_{zz}
 &=2e^{4F}\left[4s(2G+sG_s)^2+s^2G_z^2\right],
 \label{eq:F-equation}\\
 4sG_{ss}+12G_s+G_{zz}
 &=-16sF_sG_s-4F_zG_z-32F_sG,
 \label{eq:Psi-equation}
\end{align}
and 
\begin{equation}
 2\alpha_s
 =4sF_s^2-F_z^2+4se^{4F}(2G+sG_s)^2 -s^2e^{4F}G_z^2,
 \quad \alpha(0,z)=\mathbf b_i.
 \label{eq:A-s-equation}
\end{equation}

\medskip
\noindent\emph{Step 1:}
Setting $s=0$ in \eqref{eq:F-equation} gives
\begin{equation}\label{eq:Fs-fpp}
 F_s(0,0)=-\frac14f''(0).
\end{equation}
Along $z=0$, smoothness and $f'(0)=0$ yield
\begin{align*}
 F_s(s,0)=F_s(0,0)+O(s),\quad
 F_z(s,0)=sF_{sz}(0,0)+O(s^2),
\end{align*}
and 
\begin{align*}
 G(s,0)=g(0)+O(s),\quad
 G_z(s,0)&=g'(0)+O(s).
\end{align*}
Substitution into \eqref{eq:A-s-equation}, followed by integration in $s$, yields
\begin{equation}\label{eq:A-second-order}
 \alpha(s,0)=\mathbf b_i+s^2\left[
 \frac{1}{16}f''(0)^2+4e^{4f(0)}g(0)^2
 \right]+O(s^3).
\end{equation}
Hence \eqref{eq:strict-lower-bound} follows immediately unless
\begin{equation}\label{eq:fully-degenerate-second}
 f''(0)=0,
 \quad g(0)=0.
\end{equation}
Thus, we now assume \eqref{eq:fully-degenerate-second} for the remainder of the proof.

\medskip
\noindent\emph{Step 2:}
For $n\geq3$ and $m\geq0$, let $\cZ_m^{(n)}$ be the space of zonal homogeneous polynomials $P$ of degree $m$ on
\[
 \R^{n-1}_x\times\R_z
\]
which are invariant under $\SO(n-1)$ in the $x$ variables and harmonic in $\R^n$.  Every such polynomial has the form
\begin{equation}\label{eq:zonal-polynomial-form}
 P(x,z)=\sum_{j=0}^{\lfloor m/2\rfloor}
 a_j\rho^{2j}z^{m-2j},
\end{equation}
where $\rho=|x|$.
Since
\[
 \Delta_x \rho^{2(j+1)}
 =2(j+1)(2j+n-1)\rho^{2j},
\]
comparison of the coefficients in $\Delta_{\R^n}P=0$ produces
\begin{equation}\label{eq:zonal-recurrence-general}
 2(j+1)(2j+n-1)a_{j+1}
 +(m-2j)(m-2j-1)a_j=0.
\end{equation}
Thus $a_0$ determines every coefficient.  Conversely, the recurrence produces a harmonic polynomial for every choice of $a_0$; at the final index the remaining power of $z$ is zero or one, and no additional condition occurs.  Therefore restriction to the $z$-axis is an isomorphism
\begin{equation}\label{eq:zonal-restriction-iso}
 \cZ_m^{(n)}\longrightarrow\operatorname{span}\{z^m\},
 \quad P\longmapsto P(0,z).
\end{equation}
In particular, $\cZ_m^{(n)}$ is one-dimensional.  Let $H_m^{(n)}$ denote its unique element normalized by
\begin{equation}\label{eq:H-normalization}
 H_m^{(n)}(0,z)=z^m.
\end{equation}
Writing $r=(\rho^2+z^2)^{1/2}$, the two cases used below are
\begin{align}
 H_m^{(3)}(\rho,z)=r^mP_m(z/r),\quad
 H_m^{(7)}(\rho,z)=
 r^m\frac{C_m^{5/2}(z/r)}{C_m^{5/2}(1)},
 \label{eq:H7-Gegenbauer}
\end{align}
where $P_m$ and $C_m^{5/2}$ are the Legendre and Gegenbauer polynomials.  The explicit formulas are not essential; what matters is \eqref{eq:zonal-restriction-iso}.  We shall also use the consequence of the recurrence~\eqref{eq:zonal-recurrence-general} that, for even $m$,
\begin{equation}\label{eq:H-even-equator-nonzero}
 H_m^{(n)}(\rho,0)=\lambda_m^{(n)}\rho^m
 \text{ with }\lambda_m^{(n)}\neq0,
 \quad
 \partial_z H_m^{(n)}(\rho,0)=0.
\end{equation}

\noindent\emph{Step 3:}
Suppose first that $v$ is constant. Then $U$ is an axisymmetric harmonic function on $\R^3\setminus\Gamma$. This is the static case treated in~\cite{BachWeyl1922,BachWeyl2012Reprint} for two black holes, and generalized to any number of black holes in~\cite{IsraelKhan1964}. The explicit formula for $\alpha$ in this case shows~\eqref{eq:strict-lower-bound}.
For the remainder of the proof assume that $v$ is not constant.

\medskip
\noindent\emph{Step 4:} 
We study the expansion of $G$. 
Let $\widehat G(X,z)$ be the smooth seven-dimensional lift of $G(\rho^2,z)$ constructed in the proof of Proposition~\ref{prop:axis-factorization}.  
We first assume that $\widehat G(X,z)$ vanishes up to infinite order at the origin. Recall that $\widehat G(X,z)$ satisfies \eqref{eq:psi-lifted}. The standard unique continuation \cite{Aronszajn1957} yields $\widehat G\equiv0$ in a neighborhood of the origin.  Hence $v-c_i=\rho^4 G$ vanishes on a nonempty open subset of $\{\rho>0\}$.  However, $v-c_i$ satisfies the linear uniformly elliptic equation
\begin{equation}\label{eq:w-unique-continuation}
 \Delta (v-c_i)+4\grad u\cdot\grad (v-c_i)=0.
\end{equation}
The unique continuation then forces $v-c_i\equiv0$, so $v$ is globally constant, contrary to the hypothesis.  
We conclude that $\widehat G(X,z)$ does not vanish up to infinite order at the origin. 

With $\widehat G(0,0)=g(0)=0$, we assume
$$\widehat G=Q_d+O(r^{d+1}),$$
for some $d\ge 1$ and nonzero homogeneous $\SO(6)$-invariant polynomial $Q_d$. Moreover, $Q_d$ is harmonic.  Indeed, this is immediate for $d=1$, and for $d\geq2$, the contributions generated by $Q_d$ in \eqref{eq:psi-lifted} have lowest possible homogeneous degrees $d-2$, $d-1$, and $d$, 
from the Laplacian, first-order term, and zeroth-order term, respectively. It follows that $Q_d\in\cZ_d^{(7)}$, and hence $Q_d=\sigma_dH_d^{(7)}$ 
for some $\sigma_d\neq0$. 
Taylor's theorem, with one differentiated remainder, yields
\begin{equation}\label{eq:psi-leading-mode}
 \left|D^\ell\bigl(h-\sigma_dH_d^{(7)}\bigr)\right|
 \leq Cr^{d+1-\ell},
 \qquad \ell=0,1,
\end{equation}
where derivatives are understood in the lifted variables and then restricted to the orbit half-plane.

\medskip
\noindent{\it Step 5:} 
We study the expansion of $U$ and examine the validity of the estimates
\begin{equation}\label{eq:U-case3-order}
 U=f(0)+O(r^L),
 \qquad |\grad U|=O(r^{L-1}).
\end{equation}
It is possible that $U$ vanishes up to infinite order at the origin. 
In this case, \eqref{eq:U-case3-order} holds for any $L\ge 1$.
We now assume that $U$ vanishes up to a finite order at the origin, i.e., 
\begin{equation}\label{eq:U-step6-order}U=f(0)+U_k+O(r^{k+1}),\end{equation}
for some $k\ge 1$ and nonzero homogeneous polynomial $U_k$. 
The expansion here is in the original three-dimensional variables. 
In this case, \eqref{eq:U-case3-order} holds for any $1\le L\le k$.


We write the equation \eqref{eq:U-psi-equation} as 
$$\Delta_{\R^3} U=S,$$
where 
\begin{align}\label{eq:def-source-S}
 S
 =2e^{4U}\left[\rho^2(4h+\rho h_\rho)^2
 +\rho^4 h_z^2\right].
\end{align}
A simple comparison gives
\begin{equation}\label{eq:homogeneous-Poisson}
 \Delta_{\R^3}U_k= S_{k-2},
\end{equation}
where $S_{k-2}$ is the homogeneous polynomial of degree $k-2$ in the Taylor expansion of $S$. 
In the following, we only need the case $k<2d+4$. 
By \eqref{eq:psi-leading-mode}, we have $S=O(r^{2d+2})$. 
Hence, if $k<2d+4$, then $S_{k-2}=0$ and thus $\Delta_{\R^3}U_k=0$. 
As in Step 4, we have $U_k=c_kH_k^{(3)}$ for some $c_k\neq 0$. Therefore, 
$$U=f(0)+c_kH_k^{(3)}+O(r^{k+1}).$$
By assumption, $U(0,z)$ has a maximum at $z=0$. 
With $H_k^{(3)}(0,z)=z^k$, we conclude that $k$ is even and $c_k<0$. 

In conclusion, if $k$ is the first nonzero Taylor order of $U(\rho, z)-f(0)$ and $k<2d+4$, then 
\begin{equation}\label{eq:U-case1}
 \left|D^\ell\bigl(U-f(0)+\kappa H_k^{(3)}\bigr)\right|
 \leq Cr^{k+1-\ell},
 \qquad \ell=0,1, 
\end{equation}
for some even $k$ and constant $\kappa>0$.

\medskip
\noindent{\it Step 6:} 
We now study the equation  
\begin{equation}
\label{eq:alpha-rho-h}
\alpha_\rho
 =\rho(U_\rho^2-U_z^2)
   +\rho^3e^{4U}\big[(4 h+\rho h_\rho)^2-\rho^2 h_z^2\big],
\end{equation}
and examine which term on the right-hand side provides the lowest degree when expanded in terms of $\rho$. 
We note that the lowest degree of $\rho$ from $h$ is $2d+3$ by \eqref{eq:psi-leading-mode},  
and the lowest degree of $\rho$ from $U$ is $2k-1$ by \eqref{eq:U-step6-order} if $k$ is finite and $k<2d+4$. 
By comparing $2d+3$ and $2k-1$, 
we consider three exhaustive cases, 
$k<d+2$, $k=d+2$, and $k>d+2$, including $k=\infty$. We point out that $k<2d+4$ in the first two cases and hence \eqref{eq:U-case1} is applicable.

\smallskip
{\it Case I: $k<d+2$.}
 Since $k$ is even, \eqref{eq:H-even-equator-nonzero} produces
\[
 H_k^{(3)}(\rho,0)=\lambda_k\rho^k,
 \quad \lambda_k\neq0,
 \quad \partial_z H_k^{(3)}(\rho,0)=0.
\]
By \eqref{eq:alpha-rho-h}, we have
\begin{align}\label{eq:alpha-rho-h-1}
 U_\rho(\rho,0)=-\kappa k\lambda_k\rho^{k-1}+o(\rho^{k-1}),\quad
 U_z(\rho,0)=o(\rho^{k-1}).
\end{align}
By \eqref{eq:psi-leading-mode}, the contribution of $h$ to $\alpha_\rho$ in \eqref{eq:alpha-rho-h} 
is $O(\rho^{2d+3})=o(\rho^{2k-1})$.  Thus \eqref{eq:alpha-rho-h} gives
\[
 \alpha_\rho(\rho,0)=\kappa^2k^2\lambda_k^2\rho^{2k-1}
 +o(\rho^{2k-1}).
\]
Since $\alpha(0,0)=\mathbf b_i$,
\begin{equation}\label{eq:A-case1}
 \alpha(\rho,0)=\mathbf b_i+\frac{k}{2}\kappa^2\lambda_k^2\rho^{2k}
 +o(\rho^{2k})>\mathbf b_i
\end{equation}
for all sufficiently small $\rho>0$.

\smallskip
{\it Case II: $k=d+2$.}
Here $d$ is even.  Set
\[
 \lambda_k=H_k^{(3)}(1,0)\neq0,
 \quad
 \mu_d=H_d^{(7)}(1,0)\neq0.
\]
At $z=0$,   we have \eqref{eq:alpha-rho-h-1} and, by \eqref{eq:psi-leading-mode},
\begin{align*}
  h=\sigma_d\mu_d\rho^d+o(\rho^d),\quad
  h_z=o(\rho^{d-1}),
\end{align*}
and 
\begin{align*} 4 h+\rho h_\rho
 =\sigma_d(d+4)\mu_d\rho^d+o(\rho^d).
\end{align*}
Hence
\[
 \alpha_\rho(\rho,0)
 =\big[
 \kappa^2k^2\lambda_k^2
 +e^{4f(0)}\sigma_d^2(d+4)^2\mu_d^2
 \big]\rho^{2d+3}
 +o(\rho^{2d+3}).
\]
Integration yields
\begin{equation}\label{eq:A-case2}
 \alpha(\rho,0)=\mathbf b_i+
 \Big[
 \frac{k}{2}\kappa^2\lambda_k^2
 +\frac{(d+4)^2}{2d+4}e^{4f(0)}\sigma_d^2\mu_d^2
 \Big]\rho^{2d+4}
 +o(\rho^{2d+4})>\mathbf b_i.
\end{equation}

\smallskip
{\it Case III: $d+2<k$, including $k=\infty$.} 
The discussion in this case is more complicated mainly because $d$ in \eqref{eq:psi-leading-mode} is not known to be even, in which case $\rho h_z$ in \eqref{eq:alpha-rho-h} has the same order as $h$ and $\rho h_\rho$.
Setting
\begin{equation}\label{eq:def-L}
 L=\min\{k,2d+4\},
\end{equation}
then, \eqref{eq:U-case3-order} holds for this $L$. 
Since $L>d+2$, for every fixed $\eta>0$,
\begin{equation}\label{eq:U-stress-negligible}
 \int_0^{\eta z}\rho|\grad U|^2 d\rho
 =O(z^{2L})=o(z^{2d+4})
 \quad \text{as } z\downarrow0.
\end{equation}

Define the homogeneous polynomial
\begin{equation}\label{eq:def-Wd}
 W_d(\xi,\tau)=\xi^4H_d^{(7)}(\xi,\tau)
\end{equation}
and, for $\eta>0$,
\begin{equation}\label{eq:def-Id}
 I_d(\eta)=\int_0^\eta\xi^{-3}
 \left[
 (\partial_\xi W_d(\xi,1))^2
 -(\partial_\tau W_d(\xi,1))^2
 \right]d\xi.
\end{equation}
The normalization $H_d^{(7)}(0,\tau)=\tau^d$ and the recurrence \eqref{eq:zonal-recurrence-general} imply, as $\xi\downarrow0$, that
\[
 H_d^{(7)}(\xi,1)=1+O_d(\xi^2),
\]
so
\[
 \partial_\xi W_d(\xi,1)=4\xi^3+O_d(\xi^5),
 \quad
 \partial_\tau W_d(\xi,1)=d\xi^4+O_d(\xi^6).
\]
Consequently,
\begin{equation}\label{eq:Id-positive}
 I_d(\eta)=4\eta^4+O_d(\eta^6)>0
\end{equation}
for every sufficiently small fixed $\eta>0$.

We now justify that \eqref{eq:def-Id} is the leading term in $\alpha(\eta z,z)-\mathbf b_i$.  Write
\[
 R= h-\sigma_dH_d^{(7)}.
\]
On $0\leq\rho\leq\eta z$, estimate \eqref{eq:psi-leading-mode} yields
\[
 |R|\leq Cz^{d+1},
 \quad |\grad R|\leq Cz^d.
\]
Since $v-c=\rho^4(\sigma_dH_d^{(7)}+R)$ and $\rho=\xi z$, homogeneity gives
\begin{align}
 v_\rho(\xi z,z)
 &=\sigma_dz^{d+3}\partial_\xi W_d(\xi,1)
 +O\bigl(z^{d+4}(\xi^3+\xi^4)\bigr),
 \label{eq:v-rho-ray}\\
 v_z(\xi z,z)
 &=\sigma_dz^{d+3}\partial_\tau W_d(\xi,1)
 +O(z^{d+4}\xi^4).
 \label{eq:v-t-ray}
\end{align}
The powers of $\xi$ in these errors make all integrals below convergent at $\xi=0$.

Because $\alpha(0,z)=\mathbf b_i$ and $e^{4u}=\rho^{-4}e^{4U}$, \eqref{eq:alpha-rho-h} implies 
\begin{equation}\label{eq:A-ray-integral}
 \alpha(\eta z,z)=\mathbf b_i+\int_0^{\eta z}
 \left\{
 \rho(U_\rho^2-U_z^2)
 +e^{4U}\rho^{-3}(v_\rho^2-v_z^2)
 \right\}d\rho.
\end{equation}
The first term in the integrand is $o(z^{2d+3})$ by \eqref{eq:U-stress-negligible}.  In the second term, the leading/error cross terms from \eqref{eq:v-rho-ray}--\eqref{eq:v-t-ray} are $O(z^{2d+4})$, and the error squares are $O(z^{2d+5})$.  Moreover, \eqref{eq:U-case3-order} produces
\[
 e^{4U}=e^{4f(0)}(1+o(1)),
\]
uniformly on the ray segment.  Substituting $\rho=\xi z$ into \eqref{eq:A-ray-integral} therefore yields
\begin{equation}\label{eq:A-case3}
 \alpha(\eta z,z)
 =\mathbf b_i+e^{4f(0)}\sigma_d^2I_d(\eta)z^{2d+4}
 +o(z^{2d+4})>\mathbf b_i,
\end{equation}
for all sufficiently small $z>0$, by \eqref{eq:Id-positive}.

Each of the three cases gives points with $\alpha>\mathbf b_i$ arbitrarily close to the origin.  
This proves \eqref{eq:strict-lower-bound}.
\end{proof}

We note that only the smoothness obtained in Proposition~\ref{prop:axis-factorization}, and finite Taylor expansions are used.  Infinite order vanishing of the lifted $v$ is treated by unique continuation.  No real analytic regularity of the harmonic map is assumed.

\section{Proof of the Main Theorem}\label{sec:proof-main}

We first verify the endpoint behavior required to apply Proposition~\ref{prop:strict-rod}.  

\begin{lemma}\label{lem:endpoint-blowdown}
For $2\leq j\leq N$, let $f_j(z)=U(0,z)$ for $z_{j-1}^+<z<z_j^-$.
Then,
\begin{align}
 f_j(z)
 \rightarrow-\infty
 \quad\text{as }z\downarrow z_{j-1}^+ \text{ and }z\uparrow z_j^-.
 \label{eq:left-end-blowdown}
\end{align}
Consequently, $f_j$ attains a maximum at an interior point of $\Gamma_j$.
\end{lemma}

\begin{proof} 
Define
\begin{equation}\label{eq:def-endpoint-weight}
 \nu_i=
 \begin{cases}
 1,&i\in\mathcal I_{\mathrm{deg}},\\[2pt]
 \frac12,&i\in\mathcal I_{\mathrm{nd}}.
 \end{cases}
\end{equation}
At the bottom endpoint $z_{j-1}^+$, first suppose that $\mathcal H_{j-1}$ is degenerate.  Approach its puncture along the north axis and use \eqref{eq:puncture-decomposition}--\eqref{eq:bar-U-i}; since the denominator in the $\ln$ in~\eqref{eq:bar-U-i} is $2(1+b_{j-1})>0$, this gives
\[
 f_j(z)=\ln(z-z_{j-1}^+)+O(1).
\]
If $\mathcal H_{j-1}$ is nondegenerate, approach its upper pole along the adjacent north axis.  Formula \eqref{eq:upper-pole-U} gives instead
\[
 f_j(z)=\frac12\ln(z-z_{j-1}^+)+O(1).
\]
In summary, we have 
\begin{align*}
 f_j(z)=\nu_{j-1}\ln(z-z_{j-1}^+)+O(1)
 \rightarrow-\infty
 \quad\text{as }z\downarrow z_{j-1}^+.
\end{align*}
The top endpoint behavior is obtained similarly; namely, 
\begin{align*}
 f_j(z)=\nu_j\ln(z_j^--z)+O(1)
 \rightarrow-\infty
 \quad\text{as }z\uparrow z_j^-.
\end{align*}  
Axis regularity gives continuity in the interior.  The two endpoint limits force an interior maximum.
\end{proof}

\begin{proof}[Proof of Theorem~\ref{thm:main}]
Let
\[
 M=\max_{1\leq j\leq N+1}\mathbf b_j
 =\max\{0,\mathbf b_2,\ldots,\mathbf b_N\}\geq0.
\]
Proposition~\ref{prop:global-bound} gives
\begin{equation}\label{eq:alpha-below-M-final}
 \alpha\leq M
 \quad\text{throughout }\{\rho>0\}.
\end{equation}
By Lemma~\ref{lem:endpoint-blowdown}, the trace of $U$ on any bounded axis rod $\Gamma_j$, for $2\le j\le N$, has an interior maximum, hence by Proposition~\ref{prop:strict-rod}, $\alpha$ cannot achieve its maximum on any bounded axis rod. Because $\alpha=0$ on the unbounded axis rods, it follows that $M=0$. Applying Proposition~\ref{prop:strict-rod} again immediately shows that for any bounded axis rod $\Gamma_j$ we must have $\mathbf{b}_j<M=0$.
\end{proof}

\end{document}